\documentclass[onecolumn]{revtex4}

\usepackage{enumerate}
\usepackage{ifthen}
\usepackage{mathrsfs} 
\usepackage{hyperref}
\usepackage{graphicx}
\usepackage{amsmath}
\usepackage{amssymb}
\usepackage{amsthm}
\usepackage{epsfig}
\usepackage{listings}

\newcommand{\be}{\begin{equation}}
\newcommand{\ee}{\end{equation}}
\newcommand{\ba}{\begin{array}}
\newcommand{\ea}{\end{array}}
\newcommand{\bea}{\begin{eqnarray}}
\newcommand{\eea}{\end{eqnarray}}

\newcommand{\cC}{{\cal C }}

\newcommand{\cN}{{\cal N }}

\newcommand{\cP}{{\cal P }}

\newtheorem{lemma}{Lemma}

\newtheorem{remark}{Remark}

\newcommand*{\cS}{\mathcal{S}}

\newcounter{protoCount}
\newcounter{protoList}
\newsavebox{\tmpbox}
\newlength{\protobox}
\newenvironment{protocol}[3]{
\addtocounter{protoCount}{1}
\noindent \begin{lrbox}{\tmpbox}
\setlength{\protobox}{0.9\textwidth}
\addtolength{\protobox}{-0.5cm}
\begin{minipage}[c]{\protobox}
\begin{bfseries} 
#2\end{bfseries}
\ifthenelse{\equal{#3}{\empty}}{}{\\ #3}
\begin{list}{\begin{bfseries}\arabic{protoList}:\end{bfseries}}
{\usecounter{protoList}}
}{
\end{list}
\end{minipage}\end{lrbox}
\fbox{\usebox{\tmpbox}}
\bigskip
}

\begin{document}
\title{How to efficiently select an arbitrary Clifford group element}
\author{Robert Koenig}
\affiliation{Institute for Quantum Computing and Department of Applied Mathematics, University of Waterloo, Waterloo ON N2L 3G1, Canada} 
\author{ John A. Smolin}
\affiliation{IBM T.J. Watson Research Center, Yorktown Heights, NY 10598, USA}

\begin{abstract}
We give an algorithm which produces a unique element of the Clifford group
on $n$~qubits ($\cC_n$) from an integer $0\le i < |\cC_n|$ (the number
of elements in the group). The algorithm involves~$O(n^3)$ operations.
It is a variant of the subgroup algorithm by Diaconis and
Shahshahani~\cite{DiaSha87} which is commonly applied to compact Lie
groups. We provide an adaption for the symplectic
group~$Sp(2n,\mathbb{F}_2)$ which provides, in addition to a
canonical mapping from the integers to group elements $g$, a
factorization of~$g$ into a sequence of at most $4n$~symplectic
transvections.  The algorithm can be used to efficiently select random
elements of  $\cC_n$ which is often useful in quantum information
theory and quantum computation. We also give an algorithm for the inverse map, indexing a group element in time~$O(n^3)$.
\end{abstract}
\maketitle

\section{Introduction}

The Clifford group (which we will define carefully below) is of great
interest in the field of quantum information and computation.  Though
the group is not universal for quantum computation
\cite{GottesmanKnill}, it is central in the field of quantum
error-correction codes \cite{Gottesman97}, and the use of
\emph{random} elements of the Clifford group has numerous
applications, from establishing bounds on quantum capacities
\cite{BDSW} to randomized benchmarking
\cite{EmersonAlickiZyczkowski2005,Knilletal2008,Magesanetal2011} to
data hiding \cite{DiVincenzoetal02}.  Most of these applications depend
on the useful fact that the uniform distribution over Clifford group elements constitutes a 2-design for the unitary group, that is, 
reproduces the second moments of a Haar-random unitary (see \cite{BDSW,DiVincenzoetal02,Dankert09}).

There are many ways of choosing a random Clifford element.  The most
straightforward is to simply write down all the elements of the group,
and then pick randomly from the list.  This quickly becomes impractical
because the cardinality of the group 
\begin{equation}
|\cC_n|=2^{n^2+2n}\prod_{j=1}^n (4^j-1)
\end{equation}
grows quickly with the number of qubits $n$~\footnote{As pointed out
  in~\cite{ozols}, this does not agree with e.g.,~\cite{Calde98},
  since~\cite{Calde98} assumes that~$\cC_n$ is generated by $H,P$ and
  $CNOT$, and these generate additional phases because $(PH)^3=e^{\pi
    i/4}$ resulting in an additional factor of~$8$.  This extra
  phase, irrelevant to quantum mechanics, is needed in order to write down
  a unitary representation of the group. 
}.  Other (approximate) methods have
been proposed: In \cite{DiVincenzoetal02} a method is given requiring time
$O(n^8)$ and producing an approximately random Clifford, and \cite{Dankert09} gives a method that produces an
$\epsilon$-approximate unitary $2$-design based on Cliffords (consisting of only $n\log 1/\epsilon$ gates).


Our method gives a canonical mapping of consecutive integers to a Clifford group
element.  Picking a random element is equivalent then to picking a
random integer of the size of the group.  We give both $O(n^4)$ and
$O(n^3)$ algorithms for computing the group element from the
associated integer. We also give a $O(n^3)$~algorithm realizing the inverse map, i.e., taking  group elements to integers. 

\subsection{The Pauli, Clifford, and Symplectic groups }

The Pauli group~ $\cP_n$  on $n$ qubits is 
generated by single-qubit Pauli operators 
$X_j=\left(\begin{matrix} 0 & 1\\ 1 & 0\end{matrix}\right),
  Y_j=\left(\begin{matrix} 0 & -i\\ i & 0\end{matrix}\right),
    Z_j=\left(\begin{matrix} 1 & 0\\ 0 & -1\end{matrix}\right)$
 acting on the $j$th qubit, for $j=1,\ldots,n$.  Consider the
 normalizer $\cN(\cP_n)=\{U\in U(2^n)\ |\ U\cP_n U^\dagger=\cP_n\}$ of
 $\cP_n$ in the group of unitaries~$U(2^n)$.
The Clifford group $\cC_n$ is this normalizer, neglecting
the global phase:  $\cC_n=\cN(\cP_n)/U(1)$.
 Any element
 $U\in \cC_n$ is uniquely determined up to a global phase by its
 action by conjugation on the generators of~$\cP_n$, {\em i.e.}  the
 list of parameters~$(\alpha,\beta,\gamma,\delta,r,s)$ where
 $\alpha,\beta,\gamma,\delta$ are $n\times n$ matrices of bits, and
 $r,s$ are $n$-bit vectors defined by
\begin{align}
U X_j U^\dagger=(-1)^{r_j}\prod_{i=1}^n X_i^{\alpha_{ji}}Z_i^{\beta_{ji}}\quad\textrm{ and }\quad U Z_j U^\dagger=(-1)^{s_j} \prod_{i=1}^n X_i^{\gamma_{ji}}Z_i^{\delta_{ji}}\ .\label{eq:unitaryrep}
\end{align}
Note that because unitaries preserve commutation relations among the
generators not all values for the matrices
$\alpha,\beta,\gamma,\delta$ are allowed.  This is what makes picking
a random element of the group nontrivial.  
By~\eqref{eq:unitaryrep}, the task of drawing a random Clifford
element can be rephrased as that of drawing from the corresponding
distribution of parameters~$(\alpha,\beta,\gamma,\delta,r,s)$
describing such an element.

Note also that given the list~$(\alpha,\beta,\gamma,\delta,r,s)$, there is
a classical algorithm for compiling a circuit implementing~$U$ which
is composed of $O(n^2/\log n)$~gates from the gate set
$\{\textrm{H},\textrm{CNOT},\textrm{P}\}$, see~\cite{AAGott04}. A
simpler and more (time-)efficient algorithm was proposed earlier
in~\cite{Gottesman97}; it essentially performs a form of Gaussian
elimination, has runtime~$O(n^3)$ and produces a circuit with~$O(n^2)$
gates.

The group $\cC_n/\cP_n$ has a particularly  
simple form: we have
\begin{align}
\cC_n/\cP_n  \cong {\rm Sp}(2n,\mathbb{F}_2)\equiv {\rm Sp}(2n)\label{eq:symplecticgrouprelation}
\end{align}
where the latter is the symplectic group on $\mathbb{F}_2^{2n}$, {\em i.e.}, the group of $2n\times 2n$ matrices~$S$ with entries in the two-element field $\mathbb{F}_2$ such that \begin{align}
S\Lambda(n) S^T=\Lambda(n) \equiv \bigoplus_{i=1}^n\left(\begin{matrix}
0 & 1 \\
1 & 0
\end{matrix}\right) \ .\label{eq:symplecticdef}
\end{align} In this expression, the block-diagonal matrix $\Lambda(n)$
defines the symplectic inner product $\langle v,w\rangle=v^T\cdot
\Lambda(n)w$ on~$\mathbb{F}_2^{2n}$.  Preservation of the symplectic
inner product~\eqref{eq:symplecticdef} is equivalent to the
preservation of commutation relations between the generators
of~$\cP_n$ when acted on by conjugation with the corresponding
unitary. Explicitly, if a representative $U\in \cC_n/\cP_n$
acts as~\eqref{eq:unitaryrep}, then the corresponding symplectic
matrix~$S$ has entries
\begin{align}
(\alpha_{j1},\beta_{j1},\ldots,\alpha_{jn},\beta_{jn})&\textrm{ in column } 2j-1\textrm{ and }\nonumber\\
(\gamma_{j1},\delta_{j1},\ldots,\gamma_{jn},\delta_{jn})&\textrm{ in column }2j\ ,\textrm{ for }j=1,\ldots,n\ .\label{eq:symplecticenc}
\end{align} 

Eq.~\eqref{eq:symplecticgrouprelation} gives an important
simplification to our algorithm, directly implying the following
lemma:
\begin{lemma}
Specifying an arbitrary element of the Clifford group is equivalent to
specifying an element of the Pauli group and also an element from the symplectic group.
\end{lemma}
Specifying an element of the Pauli group (up to an overall phase)
simply requires picking the bitstrings $r,s$, which is trivial.  We therefore
concentrate on how to specify elements from the symplectic group henceforth.

\subsection{Symplectic Gram-Schmidt procedure}
We will 
make use of a simple generalization of the
Gram-Schmidt orthogonaliztion procedure over the symplectic inner
product.  The basic step in this procedure takes as input a set of
vectors $\Omega\subset\mathbb{F}_2^{2n}$ and a vector $v\in\Omega$. If
$\langle v,f'\rangle=0$ for all $f'\in\Omega\backslash\{v\}$, the
output is the set $\Omega'=\Omega\backslash\{v\}$. Otherwise, the
output is a vector~$w\in\Omega\backslash\{v\}$ such that the pair
$(v,w)\in \cS_n$ is symplectic (that is, satisfies $\langle v,w\rangle=1$) and a set $\Omega'$ such that
\begin{enumerate}[(i)]
\item
$\Omega$ and $\Omega'\cup \{v,w\}$ span the same space, $|\Omega'|\leq |\Omega|-2$, and 
\item
$\langle v,f'\rangle=\langle w,f'\rangle=0$ for all $f'\in \Omega'$. 
\end{enumerate}
The vector $w$ and $\Omega'$ are obtained by first choosing~$w\in\Omega\backslash\{v\}$  such that  $\langle v,w\rangle=1$ and subsequently inserting the
vector
$f+\langle v,f\rangle w+\langle w,f\rangle v$
into~$\Omega'$ for each  $f\not\in\Omega\backslash\{v,w\}$.

Repeatedly picking a vector~$v$ (arbitrarily) in the resulting set~$\Omega'$ and reapplying this basic step yields a symplectic basis of the space spanned by the original set of vectors~$\Omega$. In particular, for any non-zero vector $v\in\mathbb{F}_2^{2n}$, a symplectic basis $(v_1,w_1,v_2,w_2,\ldots,v_n,w_n)$ of $\mathbb{F}_2^{2n}$, {\em i.e.}, a basis satisfying
\begin{align}
\langle v_j,w_k\rangle =\delta_{j,k}\qquad\textrm{ and }\qquad \langle v_j,v_k\rangle =\langle w_j,w_k\rangle =0\ 
\end{align}
with $v_1=v$ can be obtained  starting from $\Omega=\{v\}\cup \{e_1,\ldots,e_{2n}\}$, where  $e_1,\ldots,e_{2n}\in\mathbb{F}_2^{2n}$ are the standard basis vectors of $\mathbb{F}_2^{2n}$. The complexity of this procedure is easily seen to be~$O(n^3)$.

\subsection{The subgroup algorithm}
Our algorithm is an adaptation of a method for generating random
matrices from the classical compact Lie groups by Diaconis and
Shahshahani~\cite{DiaSha87} (also see~\cite{Mezz07} for a nice
description). In~\cite{DiaSha87},  a method for the Lie group $\mathsf{Sp}(2n,\mathbb{C})$ is given which partly relies on the fact that its group elements can be represented as $n\times n$ matrices with entries in the quaternions. In our case, we do not have this tool at our disposal since we are working over a finite field. Getting an efficient algorithm therefore requires some additional effort.

The core of these algorithms is called the {\em subgroup algorithm}, which is most easily explained for a finite group~$G$ with a nested chain of subgroups
\begin{align}
G_1\subset G_2\subset\cdots\subset G_{n-1}\subset G_n=G\ .\label{eq:towerofsubgroups}
\end{align}
In this situation, the map
\begin{align*}
G_n/G_{n-1}\times G_{n-1}/G_{n-2}\times\cdots\times G_2/G_1\times G_1&\rightarrow G\\
([g_n],[g_{n-1}],\ldots,[g_2],g_1)&\rightarrow g_ng_{n-1}\cdots g_1
\end{align*}
is an isomorphism. In particular, each $g\in G$ has a unique representation as $g_ng_{n-1}\cdots g_1$ with $[g_j]\in G_j/G_{j-1}$ for $j=2,\ldots,n$ and $g_1\in G_1$. This implies that given an element $g_j\in G_j$ representing a uniformly random coset $[g_j]\in G_j/G_{j-1}$ for every $j=2,\ldots,n$, and a uniformly chosen random element $g_1\in G_1$, we can obtain a uniformly distributed element of~$G$ by taking the product.

In our case we take $G_j = {\rm Sp}(2j)$
where the embedding ${\rm Sp}(2(j-1))\rightarrow {\rm Sp}(2j)$ is given by 
$S\mapsto\left(\begin{matrix}1 & 0\\
0 & 1 \end{matrix}\right)\oplus S$. Furthermore, it is easy to see that 
there is a one-to-one correspondence between the 
set 
\begin{align*}
\cS_n:=\{(v,w)\in\mathbb{F}_2^{2n}\times \mathbb{F}_2^{2n}\ |\ \langle v,w\rangle=1\}
\end{align*}
of symplectic pairs of vectors
and the cosets ${\rm Sp}(2n)/{\rm Sp}(2(n-1))$. More precisely, 
let $S_{v,w}\in {\rm Sp}(2n)$ be a symplectic matrix with $v$ in the first and $w$ in the second column for any symplectic pair~$(v,w)\in\cS_n$ (we show below how to find such a matrix). Then
\begin{align}
\begin{matrix}
\cS_n&\rightarrow &{\rm Sp}(2n)/{\rm Sp}(2(n-1))\\
(v,w)&\mapsto& [S_{v,w}]
\end{matrix}\label{eq:symplecticpaircosetmap}
\end{align}
establishes the claimed one-to-one correspondence~\footnote{
To show that this is well-defined, suppose that $[S_{v,w}]=[S_{v',w'}]$,
   then $S_{v,w}^{-1}S_{v',w'}\in 
{1\ 0\choose 0\ 1}
\oplus {\rm Sp}(2(n-1))$,
 and it follows immediately that the first two standard basis vectors $e_1,e_2$ are mapped identically under $S_{v,w}$ and $S_{v',w'}$, {\em i.e.}, $v=S_{v,w}e_1=S_{v',w'}e_1=v'$ and $w=S_{v,w}e_2=S_{v',w'}e_2=w'$. To show that this parameterization is injective, suppose $[S_{v,w}]\neq [S_{v',w'}]$. Then we must have $(v,w)\neq (v',w)$ since otherwise $S^{-1}_{v,w}S_{v',w'}\in 
{1\ 0 \choose 0\ 1}
\oplus {\rm Sp}(2(n-1))$, a contradiction.
 } 
between~$\cS_n$ and ${\rm Sp}(2n)/{\rm Sp}(2(n-1))$, where we write $[S]=S\cdot {\rm Sp}(2(n-1))$ for the coset represented by~$S$.

\begin{remark}
  Another way to think of the subgroup algorithm for the symplectic
  group is the following: The coset $G_n/G_{n-1}$ will simply be
  represented by a symplectic pair $(v,w)\in\cS_n$ along with an arbitrary basis for the space
  orthogonal to $v$ and $w$.  Both our algorithms will proceed by
  picking out such a symplectic pair, then repeating in the orthogonal
  space.  It is apparent that this will give the canonical mapping we
  require.  What remains is to find an efficient algorithm for
  computing $v,w$ and the orthogonal space.
 \end{remark}

\section{Algorithms}

We will give two solutions to giving a canonical mapping of integers to ${\rm
  Sp}(2n)$. The first is based on symplectic Gaussian elimination, but
has complexity~$O(n^4)$. It is mainly of didactical interest.  The
second algorithm uses symplectic transvections and achieves a
complexity~$O(n^3)$.  Note that these algorithms do not give the same
canonical mapping.  We also provide an algorithm for the inverse problem,
finding the integer associated with a member of ${\rm SP}(2n)$.

\subsection{An algorithm with runtime $O(n^4)$ based on Gaussian elimination}

We present an algorithm $\mathsf{SYMPLECTIC}(n,i)$ which produces the $i$th 
symplectic matrix $S_i\in {\rm Sp}(2n)$.  The
algorithm is described in Fig.~\ref{fourthorder}.

We analyze the algorithm step by step.  Step 1 sets $s$ to be the
number of different choices of nonzero bitstrings of length $n$ and
$k$ to be a choice of one of them based on the input $i$.  Step 2
creates the vector $v_1$ corresponding to $k$.  Step 3 computes a
basis for the symplectic space including $v_1$.  Steps 4 and 5 pick
out a $w_1'$ based on $2n-1$ bits from $i$ such that $w_1'$ can be any
vector with $\langle w_1',v_1\rangle=1$.  Step 6 defines the desired representative~$g_n$ of a
 coset $[g_n]$.  Finally, step 7 multiplies the~$g_n$ by a symplectic
of the next smaller size in the chain, if necessary,  and returns the answer.

The runtime of this algorithm is dominated by the Gram-Schmidt
procedure $O(n^3)$, which is invoked $n$ times. Hence the total complexity of
this algorithm is~$O(n^4)$. 


\begin{figure}
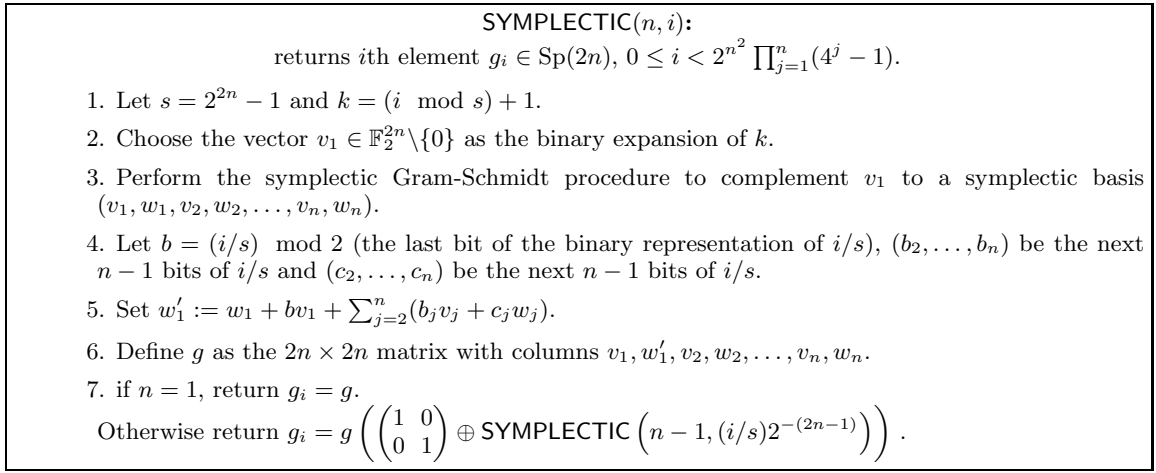

\begin{center}
\begin{protocol}{}
{$\mathsf{SYMPLECTIC}(n,i)$:}{returns $i$th element $g_i\in {\rm Sp}(2n)$,\ $0\le i <  2^{n^2} \prod_{j=1}^n(4^j-1)$.}
\item[1.] Let $s=2^{2n}-1$ and $k=(i \mod s) +1$.
\item[2.] Choose the vector $v_1\in \mathbb{F}_2^{2n}\backslash\{0\}$ as
the binary expansion of $k$.
\item[3.] Perform the symplectic Gram-Schmidt procedure  to complement~$v_1$ to a symplectic basis $(v_1,w_1,v_2,w_2,\ldots,v_n,w_n)$.
\item[4.] Let $b=(i/s) \mod 2$ (the last bit of the binary representation of
$i/s$), $(b_2,\ldots,b_n)$ be the next $n-1$ bits of $i/s$ and $(c_2,\ldots,c_n)$
be the next $n-1$ bits of $i/s$.  
\item[5.]Set $w_1':=w_1+bv_1+\sum_{j=2}^n (b_j v_j+c_jw_j)$. 
\item[6.] Define $g$ as the $2n\times 2n$ matrix with columns $v_1,w_1',v_2,w_2,\ldots,v_n,w_n$.
\item[7.] if $n=1$, return $g_i=g$.  \\
Otherwise
return  $g_i=g \left(\left(\begin{matrix}
1&0\\
0&1
\end{matrix}\right)\oplus {\mathsf{SYMPLECTIC}}\left(n-1,(i/s)2^{-(2n-1)}\right)\right)$ .
\end{protocol}
\end{center}
\caption{Symplectic algorithm with run-time $O(n^4)$.\label{fourthorder}}
\end{figure}

\subsection{An improved  algorithm with runtime $O(n^3)$}
Here we present an alternative to~$\mathsf{SYMPLECTIC}(n,i)$ which does not rely on symplectic Gaussian elimination and achieves a complexity of~$O(n^3)$. To describe and analyze our improved algorithm, we require a certain family of symplectic matrices: For a vector $h\in\mathbb{F}_2^{2n}$, define the {\em symplectic transvection} $Z_h$ as the map
\begin{align*}
Z_h: \mathbb{F}_2^{2n}&\rightarrow \mathbb{F}_2^{2n}\\
v&\mapsto v+\langle v,h\rangle h
\end{align*}
or
\begin{equation}
Z_h v = v + \langle v,h\rangle h
\end{equation}
where $Z_h$ is represented as a matrix and $v$ is a column vector.  It
is apparent that, given $h$, $Z_h v$ can be computed in $O(n)$ time,
which is faster than one could even read all $(2n)^2$ elements of
the matrix $Z_h$. Furthermore, $Z_h M$, where $M$ is a symplectic matrix,
can be computed in $O(n^2)$ time.  This will be essential to the
efficiency of our improved algorithm.

The group ${\rm Sp}(2n)$ is generated by transvections, however we do
not need this fact directly. The proof of this statement involves the
following well-known statement (see e.g.,~\cite[Section~2]{Sal06}),
which we express in an algorithmic fashion for later use.
\begin{lemma}\label{lem:transvection}
Let $x,y\in\mathbb{F}_2^{2n}\backslash\{0\}$ be two non-zero vectors. Then
\begin{align}
y=Z_h x\qquad\textrm{ for some }h\in \mathbb{F}_2^{2n}\label{eq:labelh}
\end{align}
or
\begin{align}
y=Z_{h_1}Z_{h_2} x\qquad\textrm{ for some }h_1,h_2\in \mathbb{F}_2^{2n}\label{eq:labelhtwo}
\end{align}
In other words, $x$ can be mapped to $y$ by at most two transvections.  Furthermore, there is an algorithm that outputs either~$h$ satisfying~\eqref{eq:labelh} or $(h_1,h_2)$ satisfying~\eqref{eq:labelhtwo} in time~$O(n)$.
\end{lemma}
\begin{proof}
 If $x=y$, the algorithm outputs $h=0$. Otherwise, it computes $\langle x,y\rangle$ and proceeds as follows:
\begin{enumerate}[(i)]
\item if $\langle x,y\rangle=1$, the algorithm outputs $h=x+y$.  It is easy to check that this has the required property~\eqref{eq:labelh}.
\item if $\langle x,y\rangle=0$, the algorithm computes some $z\in\mathbb{F}_2^{2n}$~such that $\langle x,z\rangle=\langle z,y\rangle=1$. Concretely, this is achieved e.g., by trying to locate an index $j\in\{1,\ldots,2n\}$ such that $(x_{2j-1},x_{2j})\neq (0,0)$ and $(y_{2j-1},y_{2j})\neq (0,0)$.  If such an index~$j$ is found, then there is a pair $(v,w)\in \mathbb{F}_2^2$ such that
$x_{2j-1}w+x_{2j}v=y_{2j-1}w+y_{2j}v=1$ and we set $z=x+ve_{2j-1}+we_{2j}$. Otherwise, there must be two distinct indices $j,k\in\{1,\ldots,2n\}$ such that $(x_{2j-1},x_{2j})\neq (0,0), (y_{2j-1},y_{2j})=(0,0)$ and $(x_{2k-1},x_{2k})= (0,0), (y_{2k-1},y_{2k})\neq (0,0)$ since $x$ and $y$ are non-zero. Then there are pairs $(v,w),(v',w')\in\mathbb{F}_2^2$ such that 
$x_{2j-1}w+x_{2j}v=y_{2k-1}w'+y_{2k}v'=1$ and we set $z=x+ve_{2j-1}+we_{2j}+v'e_{2k-1}+w'e_{2k}$.

This reduces the problem to~$(i)$ (mapping $x$ to $z$ and $z$ to $y$); the algorithm  outputs $h_1=x+z$ and $h_2=z+y$  and~\eqref{eq:labelhtwo} follows.
\end{enumerate}
\end{proof}

\begin{figure}
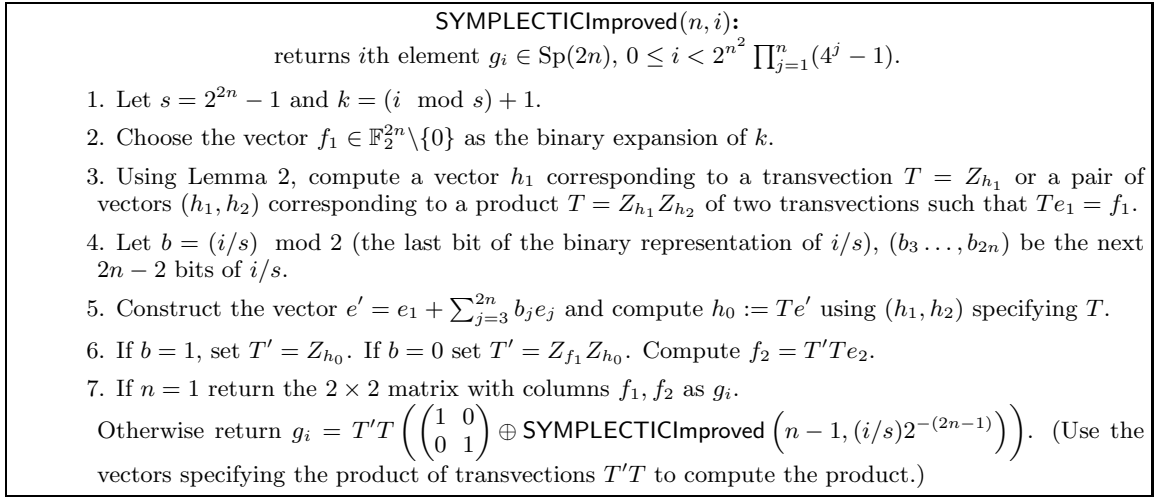

\begin{center}
\begin{protocol}{}
{$\mathsf{SYMPLECTICImproved}(n,i)$:}{returns $i$th element $g_i\in {\rm Sp}(2n)$,\ $0\le i <  2^{n^2} \prod_{j=1}^n(4^j-1)$.}
\item[1.] Let $s=2^{2n}-1$ and $k=(i \mod s)+1$.
\item[2.] Choose the vector $f_1\in \mathbb{F}_2^{2n}\backslash\{0\}$ as
the binary expansion of $k$.
\item[3.] Using Lemma~\ref{lem:transvection}, compute a vector $h_1$ corresponding to a transvection $T=Z_{h_1}$ or a pair of vectors~$(h_1,h_2)$ corresponding to  a product~$T=Z_{h_1}Z_{h_2}$ of two transvections such that $Te_1=f_1$. 
\item[4.] Let $b=(i/s) \mod 2$ (the last bit of the binary representation of
$i/s$), $(b_3\ldots,b_{2n})$ be the next $2n-2$ bits of $i/s$.
\item[5.] Construct the vector $e'=e_1+ \sum_{j=3}^{2n} b_j e_j$ and compute 
$h_0:=T e'$ using $(h_1,h_2)$ specifying $T$. 
\item[6.] If $b=1$, set $T'=Z_{h_0}$.   If $b=0$ set $T'=Z_{f_1} Z_{h_0}$.  Compute $f_2=T'T e_2$.
\item[7.] If $n=1$ return the $2\times 2$ matrix with columns $f_1,f_2$ as $g_i$.\\
Otherwise return $g_i=T'T \left(
\left(\begin{matrix}
1&0\\
0&1
\end{matrix}\right) \oplus {\mathsf{SYMPLECTICImproved}}\left(n-1,(i/s)2^{-(2n-1)}\right)\right)$. (Use the
vectors specifying the product of transvections $T'T$ to compute the product.)
\end{protocol}
\end{center}
\caption{Improved symplectic algorithm using transvections that runs in time
$O(n^3)$.
\label{improvedalgorithm}}
\end{figure} 

Our improved algorithm based on transvections is shown in
Fig.~\ref{improvedalgorithm}.  Python code that implements it
can be found in the appendix.
We now analyze it step by step.  Step 1
sets $s$ to be the number of different choices of nonzero bitstrings
of length $n$ and $k$ to be a choice of one of them based on the input
$i$.  Step 2 creates the vector $f_1$ corresponding to $k$.  So far
this is just as in the original $\mathsf{SYMPLECTIC}$, save that $v_1$
is now named $f_1$.  Step 3 computes the transvection(s)
that transform the first standard basis vector $e_1$ to
$f_1$.  This can be done efficiently using the algorithm of
Lemma~\ref{lem:transvection}. Step 4 again picks out the bits
that will specify a vector ($T'Te_2$, computed subsequently) which forms a symplectic pair with~$f_1$.  Step 5 and 6 find the transvection or pair
of transvections $T'$ with the property that $T'T e_1=f_1$ and $T'Te_2$ is an arbitrary vector forming  a symplectic pair with~$f_1$).  Thus, by~\eqref{eq:symplecticpaircosetmap}, $g_n\equiv T'T$ represents a unique coset~$[g_n]$ as required for the subgroup algorithm.

To see this it is convenient to define the vectors $f_\ell=Te_\ell$
for $\ell=\{1,\ldots, 2n\}$ corresponding to the images of the
standard basis vectors.
Observe that $(f_1,f_2,\ldots,f_{2n-1},f_{2n})$ is a symplectic
basis. We will show that
\begin{align}
T'T e_1=f_1\qquad\textrm{ and }\qquad T'T e_2=bf_1+f_2+\sum_{\ell=3}^{2n}b_\ell f_\ell\ .\label{eq:toprovezeonetwo}
\end{align}

 By linearity, the vector $h_0$ computed in step 5 of the algorithm has the form
$h_0=f_1+\sum_{k=3}^{2n} b_kf_k$. 
In particular, we get $\langle f_1,h_0\rangle=0$ and $\langle f_2,h_0\rangle=1$, which implies
\begin{align}
Z_{h_0}Te_1=Z_{h_0}f_1=f_1\qquad\textrm{ and }\qquad Z_{h_0}Te_2=Z_{h_0}f_2=f_1+f_2+\sum_{k=3}^{2n}b_kf_k\ \label{eq:toproveeqref}
\end{align}
by the definition of transvections. 
Consider the case when~$b=1$.  Then $T'T=Z_{h_0}T$ and~\eqref{eq:toproveeqref} reduces to~\eqref{eq:toprovezeonetwo}, as claimed.
On the other hand, if $b=0$, then $T'T=Z_{f_1}Z_{h_0} T$ and we can use~\eqref{eq:toproveeqref} to compute
\begin{align*}
T'Te_1&=Z_{f_1}Z_{h_0}Te_1=Z_{f_1}f_1=f_1\qquad \textrm{and}\\
T'Te_2&=Z_{f_1}Z_{h_0}Te_2=Z_{f_1}(f_1+f_2+\sum_{k=3}^{2n}b_kf_k)=f_2+\sum_{k=3}^{2n}b_kf_k\ ,
\end{align*}
confirming~\eqref{eq:toprovezeonetwo}.   

Finally, step 7 multiplies $T'T$ by a symplectic
of the next smaller size in the chain, if necessary,  and returns the answer.
This multiplication takes $O(n^2)$~time because~$T'T$  is a product of transvections associated with known vectors.
  Since there are $n$ recursions,  the total complexity is $O(n^3)$.

\subsection{An algorithm for the inverse problem with runtime $O(n^3)$}

\begin{figure}
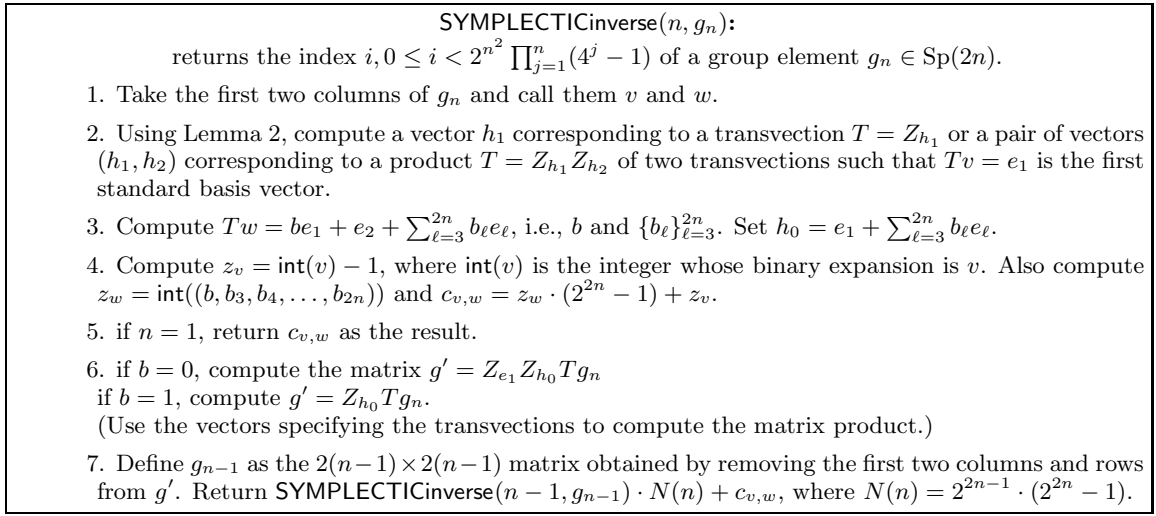

\begin{center}
\begin{protocol}{}
{$\mathsf{SYMPLECTICinverse}(n,g_n)$:}{returns the index $i, 0\leq i< 2^{n^2} \prod_{j=1}^n (4^j-1)$ of a group element $g_n\in {\rm Sp}(2n)$.}
\item[1.] Take the first two columns of $g_n$ and call them $v$ and $w$. 

\item[2.] Using Lemma~\ref{lem:transvection}, compute a vector $h_1$ corresponding to a transvection $T=Z_{h_1}$ 
or a pair of vectors $(h_1,h_2)$ corresponding to a product $T=Z_{h_1}Z_{h_2}$ of two transvections such that  $T v=e_1$ is the first standard basis vector.
\item[3.] Compute $Tw=be_1+e_2+\sum_{\ell=3}^{2n}b_\ell e_\ell$, i.e., $b$ and $\{b_\ell\}_{\ell=3}^{2n}$. Set $h_0=e_1+\sum_{\ell=3}^{2n}b_\ell e_\ell$.
\item[4.] 
\newcommand*{\mint}{\mathop{\mathsf{int}}}
Compute $z_v=\mint(v)-1$, where
$\mint(v)$ is the integer whose binary expansion is~$v$.
Also compute $z_w=\mint((b,b_3,b_4,\ldots,b_{2n}))$ and~$c_{v,w}=z_w\cdot (2^{2n}-1)+z_v$. 
\item[5.] if $n=1$, return $c_{v,w}$ as the result. 
\item[6.] if $b=0$, compute the matrix $g'=Z_{e_1}Z_{h_0}Tg_n$\\
if $b=1$, compute $g'=Z_{h_0}Tg_n$. \\
(Use the vectors specifying the transvections to compute the matrix product.)
\item[7.] Define $g_{n-1}$ as the $2(n-1)\times 2(n-1)$ matrix obtained by removing the first two columns and rows from~$g'$. Return  
$\mathsf{SYMPLECTICinverse}(n-1,g_{n-1})\cdot N(n)+c_{v,w}$, where
 $N(n)=2^{2n-1}\cdot (2^{2n}-1)$. 
\end{protocol}
\end{center}
\caption{Algorithm for taking group elements to numbers: this map implements the inverse of the map $i\mapsto \mathsf{SYMPLECTICImproved}(n,i)$.  Its runtime is~$O(n^3)$. \label{improvedalgorithminverse}}
\end{figure}

Consider the inverse problem: given a group element $g_n\in\mathrm{Sp}(2n)$, we would like to associate to it a unique index $i=i(g_n)$ where $0\leq i< |\mathrm{Sp}(2n)|=2^{n^2}\prod_{j=1}^n (4^j-1)$.  With similar reasoning as before, we can construct an efficient algorithm achieving this. It is shown in Fig.~\ref{improvedalgorithminverse} and will be referred to as $\mathsf{SYMPLECTICinverse}$. It implements the exact inverse map of the  map $i\mapsto \mathsf{SYMPLECTICImproved}(n,i)$ defined by the algorithm in Fig.~\ref{improvedalgorithm} and runs in time~$O(n^3)$. 

 Given a matrix $g_n\in \mathrm{Sp}(2n)$, the algorithm~$\mathsf{SYMPLECTICinverse}$ proceeds recursively by factorizing the given group element into representatives of cosets. Clearly, by definition of $\cS_n$, the coset in $\mathrm{Sp}(2n)/\mathrm{Sp}(2(n-1))$ can be read off from the first two columns~$(v,w)$ of~$g_n$ (Step 1).  To uniquely index different cosets, the algorithm relies on the transvection~$T$ computed in step~$2$. After step~$3$, the non-zero vector~$v$, together with the (arbitrary) bits $b,\{b_\ell\}_{\ell=3}^{2n}$, uniquely specify the symplectic pair~$(v,w)$ (and hence a coset). In step~$4$, this is used to compute an associated (unique) number $c_{v,w}$, where $0\leq c_{v,w}<N(n)$ and where $N(n)$ is the number of different cosets  in $\mathrm{Sp}(2n)/\mathrm{Sp}(2(n-1))$. 
If $n=1$, the number $c_{v,w}$ already indexes a unique group element in $\mathrm{Sp}(2)$, and no recursion is necessary (step 5). 

If $n>1$, the algorithm constructs a symplectic matrix~$V$ such that 
\begin{align}
g':=Vg_n=\left(\begin{matrix}
1 & 0\\
0 & 1
\end{matrix}
\right)\oplus g_{n-1}\ \label{gprimeproperty}
\end{align}
and returns the value $\mathsf{SYMPLECTICinverse}(n-1,g_{n-1})\cdot N(n)+c_{v,w}$ (Step 7). This number encodes both~$c_{v,w}$, i.e., the coset in $\mathrm{Sp}(2n)/\mathrm{Sp}(2(n-1))$, as well as the all the cosets in the chain of subgroups. 

It is clear that this algorithm has runtime~$O(n^3)$ if the matrix product in step~$6$ is computed using the vectors specifying the transvections. It remains to show that the matrix $g'$ constructed in step~$6$ has property~\eqref{gprimeproperty}. 

By definition, we have $g_n e_1=v$ and $g_ne_2=w$. In particular, the definition of $T$, the coefficients $b$, $\{b_\ell\}_{\ell=3}^{2n}$ and $h_0$ give
\begin{align*}
(Tg_n) e_1&=e_1\\
(Tg_n) e_2&=b e_1+e_2+\sum_{\ell=3}^{2n}b_\ell e_\ell=h_0+(b-1)e_1+e_2\ .
\end{align*}
Since $\langle e_1,h_0\rangle=0$, $\langle h_0,h_0\rangle=0$ and $\langle e_2,h_0\rangle=1$,  this implies
\begin{align*}
(Z_{h_0}Tg_n) e_1&=e_1\\
(Z_{h_0}Tg_n) e_2&=(b-1) e_1+e_2\ .
\end{align*}
This shows that if $b=1$, then $g'=Z_{h_0}Tg_n$ has the required property. If $b=0$, we use the fact that $Z_{e_1}e_1=e_1$ and $Z_{e_1}(e_1+e_2)=e_2$ to conclude that $g'=Z_{e_1}Z_{h_0}Tg_n$ has the desired form.

\section{Acknowledgments} J.A.S was supported by IARPA MQCO program under contract no. W911NF-10-1-0324.

\bibliographystyle{plain}
\bibliography{q}

\pagebreak
\appendix*
\section{Python code implementing SYMPLECTICimproved and SYMPLECTICinverse}
\begin{lstlisting}
# canonical ordering of symplectic group elements
# from "How to efficiently select an arbitrary clifford group element"
#       by Robert Koenig and John A. Smolin
#

from numpy import *
from time import clock

def directsum(m1,m2):
    n1=len(m1[0])
    n2=len(m2[0])
    output=zeros((n1+n2,n1+n2),dtype=int8)
    for i in range(0,n1):
        for j in range(0,n1):
            output[i,j]=m1[i,j]
    for i in range(0,n2):
        for j in range(0,n2):
            output[i+n1,j+n1]=m2[i,j]
    return output
######### end directsum

def inner(v,w):  # symplectic inner product
    t=0
    for i in range(0,size(v)>>1):
        t+=v[2*i]*w[2*i+1]
        t+=w[2*i]*v[2*i+1]
    return t%2

def transvection(k,v): # applies transvection Z_k to v
    return (v+inner(k,v)*k)%2

def int2bits(i,n):  # converts integer i to an length n array of bits
    output=zeros(n,dtype=int8)
    for j in range(0,n):
        output[j]=i&1
        i>>=1
    return output

    
def findtransvection(x,y):  # finds h1,h2 such that y = Z_h1 Z_h2 x 
# Lemma 2 in the text
# Note that if only one transvection is required output[1] will be 
#         zero and applying the all-zero transvection does nothing.
    output=zeros((2,size(x)),dtype=int8)
    if array_equal(x,y):
        return output
    if inner(x,y)==1:
        output[0]=(x+y)%2
        return output
#
    # find a pair where they are both not 00
    z=zeros(size(x))
    for i in range(0,size(x)>>1):
        ii=2*i
        if ((x[ii]+x[ii+1]) != 0) and ((y[ii]+y[ii+1]) != 0):  # found the pair
            z[ii]=(x[ii]+y[ii])%2
            z[ii+1]=(x[ii+1]+y[ii+1])%2
            if (z[ii]+z[ii+1])==0:  # they were the same so they added to 00
                z[ii+1]=1
                if x[ii]!=x[ii+1]:
                    z[ii]=1
            output[0]=(x+z)%2
            output[1]=(y+z)%2
            return output
    # didn't find a pair
    # so look for two places where x has 00 and y doesn't, and vice versa
#        
    # first y==00 and x doesn't
    for i in range(0,size(x)>>1):
        ii=2*i
        if ((x[ii]+x[ii+1]) != 0) and ((y[ii]+y[ii+1]) == 0):  # found the pair
            if x[ii]==x[ii+1]:
                z[ii+1]=1
            else:
                z[ii+1]=x[ii]
                z[ii]=x[ii+1]
            break
#  
    # finally x==00 and y doesn't
    for i in range(0,size(x)>>1):
        ii=2*i
        if ((x[ii]+x[ii+1]) == 0) and ((y[ii]+y[ii+1]) != 0):  # found the pair
            if y[ii]==y[ii+1]:
                z[ii+1]=1
            else:
                z[ii+1]=y[ii]
                z[ii]=y[ii+1]
            break
    output[0]=(x+z)%2
    output[1]=(y+z)%2
    return output
###################### end findtransvction

################################################################################
def symplectic(i,n): # output symplectic canonical matrix i of size 2nX2n
################################################################################
#    Note, compared to the text the transpose of the symplectic matrix
#    is returned.  This is not particularly important since
#                 Transpose(g in Sp(2n)) is in Sp(2n)
#    but it means the program doesn't quite agree with the algorithm in the 
#    text. In python, row ordering of matrices is convenient, so it is used 
#    internally, but for column ordering is used in the text so that matrix 
#    multiplication of symplectics will correspond to conjugation by 
#    unitaries as conventionally defined Eq. (2).  We can't just return the 
#    transpose every time as this would alternate doing the incorrect thing 
#    as the algorithm recurses.
#
    nn=2*n   # this is convenient to have
    # step 1
    s=((1<<nn)-1)
    k=(i%s)+1
    i/=s
#
    # step 2
    f1=int2bits(k,nn) 
#
    # step 3
    e1=zeros(nn,dtype=int8) # define first basis vectors
    e1[0]=1
    T=findtransvection(e1,f1) # use Lemma 2 to compute T
#
    # step 4
    # b[0]=b in the text, b[1]...b[2n-2] are b_3...b_2n in the text
    bits=int2bits(i%(1<<(nn-1)),nn-1)
#
    # step 5
    eprime=copy(e1)
    for j in range(2,nn):
        eprime[j]=bits[j-1]
    h0=transvection(T[0],eprime)
    h0=transvection(T[1],h0)
#
    # step 6
    if bits[0]==1:   
        f1*=0
     # T' from the text will be Z_f1 Z_h0.  If f1 has been set to zero 
     #                                       it doesn't do anything
     # We could now compute f2 as said in the text but step 7 is slightly
     # changed and will recompute f1,f2 for us anyway
#
   # step 7
#  define the 2x2 identity matrix    
    id2=zeros((2,2),dtype=int8)
    id2[0,0]=1
    id2[1,1]=1
#
    if n!=1:
        g=directsum(id2,symplectic(i>>(nn-1),n-1))
    else:
        g=id2
#
    for j in range(0,nn):
        g[j]=transvection(T[0],g[j])
        g[j]=transvection(T[1],g[j])
        g[j]=transvection(h0,g[j])
        g[j]=transvection(f1,g[j])
#
    return g
############# end symplectic

def bits2int(b,nn):  # converts an nn-bit string b to an integer between 0 and 2^n-1
    output=0
    tmp=1
    for j in range(0,nn):
        if b[j]==1:
           output=output+tmp
        tmp=tmp*2
    return output

def numberofcosets(n): # returns  the number of different cosets
  x=power(2,2*n-1)*(power(2,2*n)-1)
  return x;

def numberofsymplectic(n): # returns the number of symplectic group elements
   x=1;
   for j in range(1,n+1):
     x=x*numberofcosets(j);
   return x;


################################################################################
def symplecticinverse(n,gn): # produce an index associated with group element gn
################################################################################

  nn=2*n   # this is convenient to have

# step 1
  v=gn[0];
  w=gn[1];

# step 2
  e1=zeros(nn,dtype=int8) # define first basis vectors
  e1[0]=1
  T=findtransvection(v,e1); # use Lemma 2 to compute T

# step 3
  tw=copy(w)
  tw=transvection(T[0],tw)
  tw=transvection(T[1],tw)
  b=tw[0];
  h0=zeros(nn,dtype=int8)
  h0[0]=1
  h0[1]=0 
  for j in range(2,nn):
     h0[j]=tw[j]


# step 4
  bb=zeros(nn-1,dtype=int8)
  bb[0]=b;
  for j in range(2,nn):
    bb[j-1] =tw[j];
  zv=bits2int(v,nn)-1; # number between 0...2^(2n)-2
                                    # indexing non-zero bitstring v of length 2n

  zw=bits2int(bb,nn-1);  # number between 0..2^(2n-1)-1
                                         #indexing w (such that v,w is symplectic pair)
  cvw=zw*(power(2, 2*n)-1)+zv;
  # cvw is a number indexing the unique coset specified by (v,w)
   # it is between 0...2^(2n-1)*(2^(2n)-1)-1=numberofcosets(n)-1
  

  #step 5
  if n==1: 
     return cvw
  
  #step 6
  gprime=copy(gn);
  if b==0:
     for j in range(0,nn):
       gprime[j]=transvection(T[1],transvection(T[0],gn[j]))
       gprime[j]=transvection(h0,gprime[j])
       gprime[j]=transvection(e1,gprime[j])
  else:
    for j in range(0,nn):
       gprime[j]=transvection(T[1],transvection(T[0],gn[j]))
       gprime[j]=transvection(h0,gprime[j])

  # step 7
  gnew=gprime[2:nn,2:nn]; # take submatrix
  gnidx=symplecticinverse(n-1,gnew)*numberofcosets(n)+cvw;
  return gnidx
####### end symplecticinverse
\end{lstlisting}

\end{document}